\documentclass{elsarticle}

\usepackage{a4wide}
\usepackage{amssymb,amsmath,amsthm}
\usepackage{graphicx}
\usepackage{hyperref,color}
\usepackage[font=small,labelfont=bf]{caption}

\allowdisplaybreaks

\newtheorem{theorem}{Theorem}

\theoremstyle{definition}

\theoremstyle{lemma}
\newtheorem{lemma}{Lemma}

\makeatletter
\def\ps@pprintTitle{%
  \let\@oddhead\@empty
  \let\@evenhead\@empty
  \let\@oddfoot\@empty
  \let\@evenfoot\@oddfoot
}
\makeatother

\begin{document}
	
\begin{frontmatter}
		
\title{Epidemic Transmission Modeling with Fractional Derivatives and Environmental Pathogens}

\author[Add:b]{Moein Khalighi}
\ead{moein.khalighi@utu.fi}

\author[Add:a]{Fa\"{\i}\c{c}al Nda\"{\i}rou}
\ead{faical@math.bas.bg}

\author[Add:b]{Leo Lahti \corref{corD}}
\ead{leo.lahti@utu.fi}
\cortext[corD]{Corresponding author.}

\address[Add:a]{Institute of Mathematics and Informatics, Bulgarian Academy of Sciences, Bulgaria, Sofia 1113, ul. Akad. G. Bonchev, bl. 8}

\address[Add:b]{Department of Computing, Faculty of Technology, University of Turku}


\begin{abstract}
This research presents an advanced fractional-order compartmental model designed to delve into the complexities of COVID-19 transmission dynamics, specifically accounting for the influence of environmental pathogens on disease spread. By enhancing the classical compartmental framework, our model distinctively incorporates the effects of order derivatives and environmental shedding mechanisms on the basic reproduction numbers, thus offering a holistic perspective on transmission dynamics. Leveraging fractional calculus, the model adeptly captures the memory effect associated with disease spread, providing an authentic depiction of the virus's real-world propagation patterns. A thorough mathematical analysis confirming the existence, uniqueness, and stability of the model's solutions emphasizes its robustness. Furthermore, the numerical simulations, meticulously calibrated with real COVID-19 case data, affirm the model's capacity to emulate observed transmission trends, demonstrating the pivotal role of environmental transmission vectors in shaping public health strategies. The study highlights the critical role of environmental sanitation and targeted interventions in controlling the pandemic's spread, suggesting new insights for research and policy-making in infectious disease management.
\end{abstract}

\begin{keyword}
mathematical modeling of COVID-19 pandemics, pathogens, environmental effect 
\sep fractional differential equations
\sep numerical simulations.

\medskip

\MSC[2010]{26A33 \sep 34A08 \sep 92D30.}

\end{keyword}

\end{frontmatter}


\section{Introduction}
A comprehensive understanding of the disease transmission dynamics is essential to control the spread of epidemic diseases and prevent future viral outbreaks. Many advanced mathematical models have been developed to describe the spread of the virus. In this context, various incidence functions exist to model virus mutations during infection when infectious individuals' dynamics change. Examples include the bilinear incidence or law of mass action described in \cite{fays3} and \cite{boyer}, the saturated incidence function in \cite{MR3739726} and \cite{MR2927215}, and the Beddington--DeAngelis function in \cite{nilam} and \cite{MR3855214}. Other specific nonlinear incidence functions are also used in \cite{boukma} and \cite{lotfi}, and recent developments in this field suggest utilizing Stieltjes derivatives to model time-varying contact rates and other parameters~\cite{area}.

The existing models for studying disease transmission have faced limitations due to their simplistic assumptions of uniform interactions among individuals. Such models overlook the intricate dynamics involved in disease spread, including the transmission via contaminated environments and the movement of live viruses. To address these shortcomings, recent research has proposed new modeling approaches that include environmental pathogens~\cite{mwalili, singh2023covid, garba2020modeling}, showing promise in capturing more realistic disease dynamics. However, these advanced models often struggle to fit real-world data accurately, partly because they omit critical factors such as mortality due to the virus~\cite{mwalili}. Furthermore, the application of fractional calculus, an area with significant potential for improving model accuracy, has not been thoroughly explored in this context.

Fractional calculus represents a novel method that enhances disease transmission models by integrating memory effects and long-range dependencies~\cite{fays1, fays2, FracSIR2017, NDAIROU2023113361}. This approach, by using incommensurate fractional derivatives, offers enhanced flexibility, enabling the models to more accurately reflect the complex and diverse nature of disease spread~\cite{ndairou2023ebola, JAHANSHAHI}. This complexity encompasses varying demographics, social behaviors, and intervention strategies that influence transmission rates. Despite the application of fractional order derivatives in modeling disease dynamics with considerations for environmental factors and social distancing measures, challenges remain regarding the mathematical precision and reliability of critical indicators, such as the basic reproduction number, $R_0$~\cite{fays1}.

In our research, we tackle these challenges by ensuring precise computation of $R_0$ and affirming the local and global stability of the disease-free equilibrium. We incorporate incommensurate fractional order derivatives into ordinary differential equations (ODEs) to enhance the depiction of memory effects. This improvement aims to develop a more flexible model for assessing how past infections influence future disease transmission. Our model also includes deceased compartments. This work delves into the existence, uniqueness, and boundedness of the system's solutions within the fractional framework. Notably, in our fractional model, the value of $R_0$ is also influenced by the derivative order, providing a more detailed and accurate understanding of disease transmission dynamics. Furthermore, we conduct a comparative analysis of $R_0$ in scenarios with and without environmental pathogens, considering both integer and fractional order models, to underline the critical role of environmental transmission pathways.

The article is organized as follows: In Section \ref{sec:model}, we introduce the structure of our model. It features Caputo fractional derivatives and delineates various population classes—susceptible, exposed, symptomatic, and asymptomatic infectious individuals, those in recovery, fatalities, and a compartment for environmental pathogens, each with its respective coefficients. Section \ref{Sec:mathAnalis} delves into the mathematical analysis of the model. Here, we explore the existence and uniqueness of solutions, their boundedness, the derivation of the basic reproduction number, and an examination of the disease-free equilibrium to assess both local and global stability. Section \ref{sec:numerical} presents the numerical findings, including the fitting of the model to real data from South Africa. Additionally, sensitivity analysis is performed to evaluate the effect of various model parameters and order of derivatives on the basic reproduction number, which offers valuable insights for controlling the spread of COVID-19.

\section{The fractional-order model}
\label{sec:model}

This section presents a modified fractional compartmental model for COVID-19, building upon the framework originally proposed by \cite{mwalili}. The following system of equations describes the model:

\begin{equation}
\label{model}
\begin{cases}
\displaystyle{ ^{C}D^{\alpha_S}_{0+}S(t) 
= \Lambda^{\alpha_{S}} -\frac{\beta_1^{\alpha_S} S(t)W(t)}{1+\phi_1 W(t)}-\frac{\beta_2^{\alpha_S}S(t)\left(I_A(t)+ I_S(t)\right)}{1+\phi_2\left(I_A(t)+I_S(t)\right)} + \psi^{\alpha_S} E(t) -\mu^{\alpha_S} S(t)},\\[3mm]
\displaystyle{ ^{C}D^{\alpha_E}_{0+}E(t)
= \frac{\beta_1^{\alpha_E} S(t)W(t)}{1+\phi_1 W(t)}+\frac{\beta_2^{\alpha_E}S(t)\left(I_A(t)+ I_S(t)\right)}{1+\phi_2\left(I_A(t)+I_S(t)\right)}-\psi^{\alpha_E} E(t) -\mu^{\alpha_E} E(t) -\omega^{\alpha_E} E(t)}, \\[3mm]
\displaystyle{ ^{C}D^{\alpha_{I_A}}_{0+}I_{A}(t)
= (1-\delta)\omega^{\alpha_{I_A}} E(t) - (\mu^{\alpha_{I_A}} + \sigma^{\alpha_{I_A}})I_A(t) - \gamma_A^{\alpha_{I_A}} I_A(t)}, \\[3mm]
\displaystyle{ ^{C}D^{\alpha_{I_S}}_{0+}I_{S}(t)
= \delta \omega^{\alpha_{I_S}} E(t) - (\mu^{\alpha_{I_S}} + \sigma^{\alpha_{I_S}})I_S(t) -\gamma_S^{\alpha_{I_S}} I_S(t)}, \\[3mm]
\displaystyle{^{C}D^{\alpha_R}_{0+}R(t)
= \gamma_S^{\alpha_{R}} I_S(t) + \gamma_A^{\alpha_{R}} I_A(t) -\mu^{\alpha_{R}} R(t) },\\[3mm]
\displaystyle{^{C}D^{\alpha_D}_{0+}D(t)
= \sigma^{\alpha_{D}} (I_S(t) + I_A(t)) -\mu^{\alpha_{D}} D(t) },\\[3mm]
\displaystyle{^{C}D^{\alpha_W}_{0+}W(t)
= \eta_A^{\alpha_{W}} I_A(t) + \eta_S^{\alpha_{W}} I_S(t)-\mu_P^{\alpha_{W}}W(t)}.
\end{cases}
\end{equation}
The model delineates the dynamics of COVID-19 across seven distinct compartments, representing the stages of the disease from susceptibility to recovery or death, alongside the environmental presence of pathogens. These compartments are denoted as follows: Susceptible individuals ($S$), Exposed individuals ($E$), Asymptomatic infectious individuals ($I_A$), Symptomatic infectious individuals ($I_S$), Recovered individuals ($R$), Deceased individuals ($D$), and Environmental pathogens ($W$). Detailed parameter descriptions are provided in Table \ref{tab:param}.

\begin{table}[ht!]
\centering
\caption{Description of the parameters and their units.} 
\begin{tabular}{|l|l|l|} \hline
\textbf{Parameter} & \textbf{Description} & \textbf{Unit} \\ \hline
$\Lambda$ & Rate of new susceptible births per day & $1/\text{days}$ \\
$\mu$ & Daily per capita death rate, excluding disease deaths & $1/\text{days}$ \\
$\mu_P$ & Daily death rate of pathogens in the environment & $1/\text{days}$ \\
$\phi_1$ & Fraction interacting with infectious environments daily & - \\
$\phi_2$ & Fraction interacting with infectious individuals daily & - \\
$\beta_1$ & Daily infection rate, $S$ to $E$, via environmental contact & $1/\text{days}$ \\
$\beta_2$ & Daily infection rate, $S$ to $E$, via contact with $I_A$ or $I_S$ & $1/\text{days}$ \\
$\delta$ & Fraction of symptomatic infectious people & - \\
$\psi$ & Daily rate of $E$ to $S$ due to immune response & $1/\text{days}$ \\
$\omega$ & Daily progression rate from $E$ to $I_A$ or $I_S$ & $1/\text{days}$ \\
$\sigma$ & Daily disease mortality rate & $1/\text{days}$ \\
$\gamma_S$ & Daily recovery rate of symptomatic individuals & $1/\text{days}$ \\
$\gamma_A$ & Daily recovery rate of asymptomatic individuals & $1/\text{days}$ \\
$\eta_S$ & Daily virus spread rate to environment by $I_S$ & $1/\text{days}$ \\
$\eta_A$ & Daily virus spread rate to environment by $I_A$ & $1/\text{days}$ \\
\hline
\end{tabular}
\label{tab:param}
\end{table}

In our model, Caputo fractional-order derivatives are employed to capture the complex dynamics of the disease's spread, with a formal definition provided in the subsequent section. It is important to note that all model parameters, except for $\phi_1$, $\phi_2$, and $\delta$, possess dimensions of $1/t^{\alpha_i}$ (where $i$ represents the different compartments: $S$, $E$, $I_A$, $I_S$, $R$, $D$, $W$). To ensure dimensional consistency, as stressed by \cite{diethelm2013fractional}, these parameters are exponentiated to their respective derivative orders. This approach not only adheres to the mathematical rigor required for the model's integrity but also enhances its applicability in capturing the behaviors of disease transmission and progression.

The infection transmission mechanisms, represented by $\displaystyle{ \frac{\beta_1^{\alpha_S} SW}{1+\phi_1 W}, \; \frac{\beta_2^{\alpha_S} S(I_A+ I_S)}{1+\phi_2(I_A+I_S)}}$, integrate the concept of saturation in incidence functions, reflecting a reduction in infection rates as a result of social distancing and other behavioral adjustments (\cite{beay, MR3212235,MR2646592,MR2340347}). This approach, validated by \cite{mwalili} and further explored by \cite{razzaq}, expresses the interplay between population behavior and the spread of the disease. In our proposed model, we have also incorporated the dynamics of deceased individuals, $D$, an aspect not covered in the original model~\cite{mwalili}. This enhancement allows for a more plausible fit with the available data.

\section{Mathematical analysis of the model}\label{Sec:mathAnalis}

\subsection{Preliminaries}
This section covers key concepts from fractional calculus, including the Riemann-Liouville fractional integral, Caputo fractional derivative operator, and Mittag-Leffler function~\cite{podlubny1998fractional, kilbas2006theory}.

We denote the time-fractional Caputo derivative of order \(\alpha \in \mathbb{R}^+\), starting from the initial time \(t=0\), as \({}^C{D}_{0+}^{\alpha}\). For a continuously differentiable function \(f\), the Caputo fractional derivative at time \(t\) is given by:
\begin{equation}\label{eq:caputo}
{}^C{D}_{0+}^{\alpha} f(t) = I^{m-\alpha}_{0+}f^{(m)}(t)=\frac{1}{\Gamma (m- \alpha  )}\int_{0}^{t}{\frac{f^{(m)}(\tau )d\tau }{{(t - {\tau})}^{1+\alpha-m}}},
\end{equation}
    where \( f^{(m)}(t) \) signifies the \( m \)-th order derivative of \( f(t) \), with \( m \) being the smallest integer not less than \(\alpha\).  For the scope of this study, considering \(0 < \alpha \leq 1\), it follows that \(m=1\). The operator \(I^{1-\alpha}_{0+}\) represents the Riemann-Liouville fractional integral of order \(1-\alpha\), defined as
\begin{equation}\label{eq:RLInt}
I^{\alpha}_{0+}f(t)=\frac{1}{\Gamma (\alpha )}\int_{0}^{t}{\frac{f(\tau )d\tau }{{(t - {\tau})}^{1-\alpha}}},
\end{equation}
with \(\Gamma\) denoting the gamma function. 

The Mittag-Leffler function plays a crucial role in our discussion. This function is useful for its properties in fractional calculus and its applications. In this study, we use an asymptotic expansion of this function defined as follows:

\begin{equation*}
E_{\alpha, \beta}(z) \approx - \sum_{k=1}^{N}\frac{z^{-k}}{\Gamma(\beta-\alpha k)} + O\left(\frac{1}{|z|^{1+N}}\right),
\end{equation*}

for large values of \(|z|\) (specifically, as \(|z|\) approaches infinity), and when the argument of \(z\) (\(|\text{arg}(z)|\)) is between \(\frac{\alpha \pi}{2}\) and \(\pi\). The notation \(O\left(\frac{1}{|z|^{1+N}}\right)\) indicates the asymptotic behavior of the function as \(|z|\) becomes very large, providing insight into how the function decays over large distances or times~\cite{sulami}. 

\subsection{Existence and Uniqueness}
This section delves into the theoretical foundation concerning the existence and uniqueness of solutions for the system outlined in Equation~\ref{model}, leveraging the principles of fixed-point theory and the Picard-Lindelöf theorem~\cite{math11030555}. We reformulate system~\ref{model} as follows:
\begin{equation}
\label{short-model}
\begin{cases}
 ^{C}D^{\alpha_S}_{0+}S(t) 
= F_1(t,S),\\
^{C}D^{\alpha_E}_{0+}E(t)
= F_2(t,E),\\
 ^{C}D^{\alpha_{I_A}}_{0+}I_{A}(t)
= F_3(t,I_A), \\
 ^{C}D^{\alpha_{I_S}}_{0+}I_{S}(t)
= F_4(t,I_S), \\
^{C}D^{\alpha_R}_{0+}R(t)
= F_5(t,R),\\
^{C}D^{\alpha_D}_{0+}D(t)
= F_6(t,D),\\
^{C}D^{\alpha_W}_{0+}W(t)
= F_7(t,W),
\end{cases}
\end{equation}
where the function $F_i$, $1\leq i \leq 7$ , are the derivative functions on right side of system~\ref{model}.

Thus, our approach unfolds in two steps. Initially, we establish positive values for the initial conditions. Subsequently, we apply the fractional order integral operator, which reconstitutes the system's equations as follows:

\begin{equation}
\begin{cases}
S(t)-S(0)
= I^{\alpha_S}_{0+}\left(\Lambda^{\alpha_{S}} -\frac{\beta_1^{\alpha_S} S(t)W(t)}{1+\phi_1 W(t)}-\frac{\beta_2^{\alpha_S}S(t)\left(I_A(t)+ I_S(t)\right)}{1+\phi_2\left(I_A(t)+I_S(t)\right)} + \psi^{\alpha_S} E(t) -\mu^{\alpha_S} S(t)\right),\\
E(t)-E(0)
= I^{\alpha_E}_{0+} \left(\frac{\beta_1^{\alpha_E} S(t)W(t)}{1+\phi_1 W(t)}+\frac{\beta_2^{\alpha_E}S(t)\left(I_A(t)+ I_S(t)\right)}{1+\phi_2\left(I_A(t)+I_S(t)\right)}-\psi^{\alpha_E} E(t) -\mu^{\alpha_E} E(t) -\omega^{\alpha_E} E(t)\right), \\
I_A(t)-I_A(0)
= I^{\alpha_{I_A}}_{0+} \left( (1-\delta)\omega^{\alpha_{I_A}} E(t) - (\mu^{\alpha_{I_A}} + \sigma^{\alpha_{I_A}})I_A(t) - \gamma_A^{\alpha_{I_A}} I_A(t) \right),\\
 I_S(t)-I_S(0)= I^{\alpha_{I_S}}_{0+} \left(\delta \omega^{\alpha_{I_S}} E(t) - (\mu^{\alpha_{I_S}} + \sigma^{\alpha_{I_S}})I_S(t) -\gamma_S^{\alpha_{I_S}} I_S(t)\right),\\
R(t)-R(0)= I^{\alpha_R}_{0+} \left(\gamma_S^{\alpha_{R}} I_S(t) + \gamma_A^{\alpha_{R}} I_A(t) -\mu^{\alpha_{R}} R(t)\right),\\
D(t)-D(0) =I^{\alpha_E}_{0+} \left(\sigma^{\alpha_{D}} (I_S(t) + I_A(t)) -\mu^{\alpha_{D}} D(t)\right),\\
W(t)-W(0)= I^{\alpha_W}_{0+} \left(\eta_A^{\alpha_{W}} I_A(t) + \eta_S^{\alpha_{W}} I_S(t)-\mu_P^{\alpha_{W}}W(t)\right).
\end{cases}
\end{equation}

Utilizing Equation~\ref{short-model}, we derive the state variables expressed in terms of $F_i$ as follows:

\begin{equation}\label{eq11}
\begin{cases}
S(t)=S(0) + \frac{1}{\Gamma(\alpha_S)}\int_0^t {(t-\tau)^{\alpha_S-1}F_1(\tau,S(\tau))d\tau},\\
E(t)=E(0) + \frac{1}{\Gamma(\alpha_E)}\int_0^t {(t-\tau)^{\alpha_E-1}F_2(\tau,E(\tau))d\tau},\\
I_A(t)=I_A(0) + \frac{1}{\Gamma(\alpha_{I_A})}\int_0^t {(t-\tau)^{\alpha_{I_A}-1}F_3(\tau,I_A(\tau))d\tau},\\
I_S(t)=I_S(0) + \frac{1}{\Gamma(\alpha_{I_S})}\int_0^t {(t-\tau)^{\alpha_{I_S}-1}F_4(\tau,I_S(\tau))d\tau},\\
R(t)=R(0) + \frac{1}{\Gamma(\alpha_R)}\int_0^t {(t-\tau)^{\alpha_R-1}F_5(\tau,R(\tau))d\tau},\\
D(t)=D(0) + \frac{1}{\Gamma(\alpha_D)}\int_0^t {(t-\tau)^{\alpha_D-1}F_6(\tau,D(\tau))d\tau},\\
W(t)=W(0) + \frac{1}{\Gamma(\alpha_W)}\int_0^t {(t-\tau)^{\alpha_W-1}F_7(\tau,W(\tau))d\tau}.\\
\end{cases}
\end{equation}

Therefore, by applying the Picard iteration to Equation~\eqref{eq11}, we obtain the subsequent equations:

\begin{equation}
\begin{cases}
S_{n+1}(t)=S(0) + \frac{1}{\Gamma(\alpha_S)}\int_0^t {(t-\tau)^{\alpha_S-1}F_1(\tau,S(\tau))d\tau},\\
E_{n+1}(t)=E(0) + \frac{1}{\Gamma(\alpha_E)}\int_0^t {(t-\tau)^{\alpha_E-1}F_2(\tau,E(\tau))d\tau},\\
I_{A_{n+1}}(t)=I_A(0) + \frac{1}{\Gamma(\alpha_{I_A})}\int_0^t {(t-\tau)^{\alpha_{I_A}-1}F_3(\tau,I_A(\tau))d\tau},\\
I_{S_{n+1}}(t)=I_S(0) + \frac{1}{\Gamma(\alpha_{I_S})}\int_0^t {(t-\tau)^{\alpha_{I_S}-1}F_4(\tau,I_S(\tau))d\tau},\\
R_{n+1}(t)=R(0) + \frac{1}{\Gamma(\alpha_R)}\int_0^t {(t-\tau)^{\alpha_R-1}F_5(\tau,R(\tau))d\tau},\\
D_{n+1}(t)=D(0) + \frac{1}{\Gamma(\alpha_D)}\int_0^t {(t-\tau)^{\alpha_D-1}F_6(\tau,D(\tau))d\tau},\\
W_{n+1}(t)=W(0) + \frac{1}{\Gamma(\alpha_W)}\int_0^t {(t-\tau)^{\alpha_W-1}F_7(\tau,W(\tau))d\tau}.\\
\end{cases}
\end{equation}

To introduce the result, define $\textbf{X}(t) = (S(t), E(t), I_A(t), I_S(t), R(t), D(t), W(t))^{\intercal}$ as the state vector and 

\noindent $\textbf{F}(t, \textbf{X}(t)) = (F_1(t, S(t)), F_2(t, E(t)), F_3(t, I_A(t)), F_4(t, I_S(t)), F_5(t, R(t)), F_6(t, D(t)), F_7(t, W(t)))^{\intercal}$ as the vector of functions describing the system's dynamics. Furthermore, let $\mathbb{R}^7_+ = \{\textbf{X} \in \mathbb{R}^7 : \textbf{X} \geq 0\}$ representing the space of feasible states.

\begin{lemma}\label{lemm: mean}
    Let $f(t)$ be a function such that $f(t) \in C([0,T])$ and its Caputo fractional derivative $^{C}D^{\alpha}_{0+}f(t) \in C((0,T])$, with $0 < \alpha \leq 1$. Then, the following relation holds:
    \begin{equation}
        f(t) = f(0) + \frac{1}{\Gamma(\alpha)} \, ^{C}D^{\alpha}_{0+}f(\xi) \, (t - 0)^{\alpha},
    \end{equation}
    where $0 \leq \xi \leq t$ for all $t \in (0,T]$. This relation represents the fractional generalization of the Mean Value Theorem, as established in~\cite{ODIBAT2007286}.
\end{lemma}

\begin{lemma}
Assuming the existence of a solution for system \eqref{model}, said solution will maintain non-negativity provided that the initial conditions are also non-negative.
\end{lemma}

\begin{proof}
    To establish the result, we initially note the validity of the following conditions:
    \begin{equation}
        \begin{cases}
 ^{C}D^{\alpha_S}_{0+}S(t) |_{S=0}= \Lambda^{\alpha_S} + \psi^{\alpha_S} E(t)\geq 0,\\
^{C}D^{\alpha_E}_{0+}E(t)|_{E=0}
= \frac{\beta_1^{\alpha_E} S(t) W(t)}{1+\phi_1 W(t)}+\frac{\beta_2^{\alpha_E} S(t) (I_A(t) +I_S(t))}{1+\phi_2(t)(I_A(t) +I_S(t))}\geq 0,\\
 ^{C}D^{\alpha_{I_A}}_{0+}I_{A}(t) |_{I_A=0}
= (1-\delta)\omega^{\alpha_{I_A}} E(t) \geq 0, \\
 ^{C}D^{\alpha_{I_S}}_{0+}I_{S}(t) |_{I_S=0}
= \delta \omega^{\alpha_{I_S}} E(t) \geq 0, \\
^{C}D^{\alpha_R}_{0+}R(t) |_{R=0}
= \gamma_S^{\alpha_R} I_S(t) + \gamma_A^{\alpha_R} I_A(t) \geq 0,\\
^{C}D^{\alpha_D}_{0+}D(t) |_{D=0}
= \sigma^{\alpha_D}(I_S(t)+I_A(t)) \geq 0,\\
^{C}D^{\alpha_W}_{0+}W(t) |_{W=0}
= \eta_A^{\alpha_W} I_A(t) +\eta_S^{\alpha_W} I_S(t) \geq 0,
\end{cases}
    \end{equation}
for all $t \in [0, T]$. Given these conditions and referencing Lemma \ref{lemm: mean}, it follows that the solution vector $\textbf{X}(t) = (S(t), E(t), I_A(t), I_S(t), R(t), D(t), W(t))^{\intercal}$ of system \eqref{model} is contained within $\mathbb{R}^7_+$. This observation conclusively verifies the non-negativity of the solution, thereby completing the proof.
\end{proof}

Now, we endeavor to present and validate a pivotal outcome of our study. This outcome lays the foundational groundwork for establishing the existence and uniqueness of solutions for the system outlined in Equation \eqref{model}, details of which will be elaborated upon subsequently.

\begin{lemma}\label{lemm: lips}
    The function $\textbf{F}(t,\textbf{X}(t))$ fulfills the Lipschitz conditions, specifically:
    \begin{equation}
        \|\textbf{F}(t, \textbf{X}(t)) - \textbf{F}(t, \textbf{X}^*(t))\| \leq \Sigma \|\textbf{X}(t)-\textbf{X}^*(t)\|,
    \end{equation}
    where
    \begin{equation}
    \begin{split}
        \Sigma=max\{|\frac{\beta_1^{\alpha_S}}{\phi_1}+\frac{\beta_2^{\alpha_S}}{\phi_2}+\mu^{\alpha_S}|, | \psi^{\alpha_E}+\mu^{\alpha_E}+\omega^{\alpha_E}|, |\mu^{\alpha_{I_A}}+\sigma^{\alpha_{I_A}}+\gamma_A^{\alpha_{I_A}}|, & \\ |\mu^{\alpha_{I_S}}+\sigma^{\alpha_{I_S}}+\gamma_S^{\alpha_{I_S}}|,|\mu^{\alpha_R}|,|\mu^{\alpha_D}|,|\mu_p^{\alpha_W}|\}.
    \end{split}
    \end{equation}
\end{lemma}
\begin{proof}
    Summarising that $S(t)$ and $S^*(t)$ are couple functions yields the following equality:
\begin{equation}
        \|F_1(t,S(t)) - F_1(t, S^*(t))\|=\Big\Vert\ \left( \frac{\beta_1^{\alpha_S}W(t)}{1+\phi_1 W(t)}+\frac{\beta_2^{\alpha_S}(I_A(t)+I_S(t))}{1+\phi_2(I_A(t)+I_S(t))}+\mu^{\alpha_S}\right)(S(t)-S^*(t))\Big\Vert\ .
\end{equation}
    By defining
\begin{equation}
        \Sigma_1=|\frac{\beta_1^{\alpha_S}}{\phi_1}+\frac{\beta_2^{\alpha_S}}{\phi_2}+\mu^{\alpha_S}|,
\end{equation}
we can infer the inequality:
\begin{equation}\label{eq18}
         \|F_1(t,S(t)) - F_1(t, S^*(t))\|\leq\Sigma_1\|(S(t)-S^*(t)\|.
    \end{equation}
Proceeding similarly for the other functions yields:    \begin{align}\label{eq19}
        \|F_2(t,E(t)) - F_2(t, E^*(t))\| &\leq\Sigma_2\|(E(t)-E^*(t)\| \notag \\
        \|F_3(t,I_A(t)) - F_3(t, {I_A}^*(t))\| &\leq\Sigma_3\|(I_A(t)-{I_A}^*(t)\| \notag \\
\|F_4(t,I_S(t)) - F_4(t, {I_S}^*(t))\| &\leq\Sigma_4\|(I_S(t)-{I_S}^*(t)\|  \\
        \|F_5(t,R(t)) - F_5(t, R^*(t))\| &\leq\Sigma_5\|(R(t)-R^*(t)\| \notag \\
        \|F_6(t,D(t)) - F_6(t, D^*(t))\| &\leq\Sigma_6\|(D(t)-D^*(t)\| \notag \\
        \|F_7(t,W(t)) - F_7(t, W^*(t))\| &\leq\Sigma_7\|(W(t)-W^*(t)\| \notag ,
    \end{align}
        with the $\Sigma$ values specified as: 
        \begin{align}\label{eq20}
        \Sigma_2 &= | \psi^{\alpha_E}+\mu^{\alpha_E}+\omega^{\alpha_E}| \notag ,\\
        \Sigma_3 &= |\mu^{\alpha_{I_A}}+\sigma^{\alpha_{I_A}}+\gamma_A^{\alpha_{I_A}}|\notag ,\\
        \Sigma_4 &= |\mu^{\alpha_{I_S}}+\sigma^{\alpha_{I_S}}+\gamma_S^{\alpha_{I_S}}|  ,\\
        \Sigma_5 &= |\mu^{\alpha_R}| \notag ,\\
        \Sigma_6 &= |\mu^{\alpha_D}| \notag ,\\
        \Sigma_7 &= |\mu_p^{\alpha_W}| \notag .
        \end{align}

This analysis, from Equations \eqref{eq18} to \eqref{eq20}, demonstrates that all seven functions, $F_i$, meet the Lipschitz condition, validating their properties for the system \eqref{model}.
\end{proof}

\begin{theorem}\label{theo: 1}
    Given the conditions of Lemma \ref{lemm: lips}, if the inequality    
    \begin{equation}
        \Sigma \max_{i} \frac{T^{\alpha_i}}{\Gamma(\alpha_i +1)}<1,\quad i={S,E,I_A,I_S,R,D,W}, 
    \end{equation}
    is satisfied, then the system \eqref{model} admits a unique, positive solution.
    \end{theorem}

\begin{proof}
    The solution to system \eqref{model} can be expressed in the form:
\begin{equation}
    \textbf{X}(t) = P(\textbf{X}(t)),
\end{equation}
where $P: C([0, T], \mathbb{R}^7) \to C([0, T], \mathbb{R}^7)$ denotes the Picard operator. This operator is defined as follows:

  \begin{equation}
\begin{split}
    P(\textbf{X}(t)) =\textbf{X}(0) + \int_0^t \mathrm{diag}\Bigg( & \frac{(t-\tau)^{\alpha_{S}-1}}{\Gamma(\alpha_{S})}, \frac{(t-\tau)^{\alpha_{E}-1}}{\Gamma(\alpha_{E})}, 
     \frac{(t-\tau)^{\alpha_{I_A}-1}}{\Gamma(\alpha_{I_A})}, \frac{(t-\tau)^{\alpha_{I_S}-1}}{\Gamma(\alpha_{I_S})}, \\
    & \frac{(t-\tau)^{\alpha_{R}-1}}{\Gamma(\alpha_{R})}, \frac{(t-\tau)^{\alpha_{D}-1}}{\Gamma(\alpha_{D})}, 
     \frac{(t-\tau)^{\alpha_{W}-1}}{\Gamma(\alpha_{W})} \Bigg) \mathbf{F}(\tau, \mathbf{X}(\tau)) \, d\tau.
\end{split}
\label{eq:longEquation}
\end{equation}
Simultaneously, we encounter the series of inequalities below:
\begin{equation}
    \begin{split}
        \|P(\textbf{X}(t)) - P(\textbf{X}^*(t))\|&=
        \Big\Vert\ 
        \int_0^t \mathrm{diag}\Bigg( \frac{(t-\tau)^{\alpha_{S}-1}}{\Gamma(\alpha_{S})}, \frac{(t-\tau)^{\alpha_{E}-1}}{\Gamma(\alpha_{E})},  \\ &
     \frac{(t-\tau)^{\alpha_{I_A}-1}}{\Gamma(\alpha_{I_A})}, \frac{(t-\tau)^{\alpha_{I_S}-1}}{\Gamma(\alpha_{I_S})}, 
     \frac{(t-\tau)^{\alpha_{R}-1}}{\Gamma(\alpha_{R})}, \frac{(t-\tau)^{\alpha_{D}-1}}{\Gamma(\alpha_{D})}, 
     \frac{(t-\tau)^{\alpha_{W}-1}}{\Gamma(\alpha_{W})} \Bigg) 
 \\ & 
     \times \left( \mathbf{F}(\tau, \mathbf{X}(\tau)) - \mathbf{F}(\tau, \mathbf{X}^*(\tau)) \right) \, d\tau \Big\Vert\ \\ &
     \leq
     \Big\Vert\ 
        \int_0^t \mathrm{diag}\Bigg( \frac{(t-\tau)^{\alpha_{S}-1}}{\Gamma(\alpha_{S})}, \frac{(t-\tau)^{\alpha_{E}-1}}{\Gamma(\alpha_{E})}, \\ &
     \frac{(t-\tau)^{\alpha_{I_A}-1}}{\Gamma(\alpha_{I_A})}, \frac{(t-\tau)^{\alpha_{I_S}-1}}{\Gamma(\alpha_{I_S})}, 
     \frac{(t-\tau)^{\alpha_{R}-1}}{\Gamma(\alpha_{R})}, \frac{(t-\tau)^{\alpha_{D}-1}}{\Gamma(\alpha_{D})}, 
     \frac{(t-\tau)^{\alpha_{W}-1}}{\Gamma(\alpha_{W})} \Bigg)d\tau  \Big\Vert\
 \\  &
     \times \sup_{\tau \in [0,T]} \| \mathbf{F}(\tau, \mathbf{X}(\tau)) - \mathbf{F}(\tau, \mathbf{X}^*(\tau)) \| \\ &
     \leq
     \max_{i= S, E, IA, IS , R, D, W} \int_0^t \frac{(t-\tau)^{\alpha_{i}-1}}{\Gamma(\alpha_{i})} d\tau \sup_{\tau \in [0,T]} \| \mathbf{F}(\tau, \mathbf{X}(\tau)) - \mathbf{F}(\tau, \mathbf{X}^*(\tau)) \| \\&
     \leq
      \Sigma  \max_{i= S, E, I_A, I_S , R, D, W} \frac{T^{\alpha_i}}{\Gamma(\alpha_i +1)}  \sup_{\tau \in [0,T]} \| \mathbf{X}(\tau) - \mathbf{X}^*(\tau) \|.
    \end{split}
\end{equation}
Given that $\Sigma \max_{i \in \{S, E, I_A, I_S, R, D, W\}} \frac{T^{\alpha_i}}{\Gamma(\alpha_i +1)} < 1$ for $t \leq T$, the operator $P$ is established as a contraction. Consequently, system \eqref{model} is guaranteed to have a unique solution, thereby completing the proof.

\end{proof}

\subsection{Boundedness of the Solution}
This section demonstrates the positive invariance and boundedness of the solutions within a biologically meaningful region for the model system \eqref{model}. 

\begin{theorem}
Let \(\Omega\) be the closed set defined by
\begin{equation}
\Omega = \left\{  \textbf{X}  \in \mathbb{R}^{7}_{+}: 0\leq N(t)\leq \frac{\Lambda}{\mu}, 0\leq W\leq \frac{\Lambda(\eta_A + \eta_S)}{\mu \mu_P} \right\},
\end{equation}
where \(N(t) = S(t) + E(t) + I_A(t) + I_S(t) + R(t) + D(t)\) represents the total population at time \(t\), excluding the environmental pathogen \(W\). Then, \(\Omega\) is positively invariant under the dynamics of the system \eqref{model} with commensurate orders for all \(t > 0\), meaning that any solution starting within \(\Omega\) remains in \(\Omega\) for all future times.
\end{theorem}

\begin{proof}
To establish the theorem, we initially consider the case of integer order derivatives. Define the total population at time \(t\) as \(N(t) = S(t) + E(t) + I_A(t) + I_S(t) + R(t) + D(t)\). According to system \eqref{model}, the rate of change of \(N(t)\) is given by
\begin{equation}
    N'(t) = \Lambda - \mu N(t),
\end{equation}
indicating a balance between birth and overall mortality rates. Multiplying both sides by \(e^{\mu t}\) and integrating yields
\begin{equation}
    e^{\mu t} N(t) = \frac{\Lambda}{\mu} e^{\mu t} + C,
\end{equation}
where \(C\) is a constant determined by initial conditions, specifically \(C = N_0 - \frac{\Lambda}{\mu}\) with \(N_0 = N(0)\). This leads to
\begin{equation}
    N(t) = N_0 e^{-\mu t} + \frac{\Lambda}{\mu}(1 - e^{-\mu t}),
\end{equation}
and accordingly,
\begin{equation}
    \lim \sup_{t\rightarrow \infty} N(t) = \frac{\Lambda}{\mu}.
\end{equation}

For environmental pathogens, considering \(I_A(t) + I_S(t) \leq N(t) \leq \frac{\Lambda}{\mu}\) and applying the last equation of \eqref{model}, we deduce
\begin{equation}
    W'(t) \leq (\eta_A + \eta_S) \frac{\Lambda}{\mu} - \mu_p W(t),
\end{equation}
which implies
\begin{equation}
    W(t) \leq W(0) e^{-\mu_p t} + \frac{(\eta_A + \eta_S) \Lambda}{\mu \mu_p}(1 - e^{-\mu_p t}),
\end{equation}
ensuring
\begin{equation}
    \lim \sup_{t \rightarrow \infty} W(t) \leq \frac{(\eta_A + \eta_S) \Lambda}{\mu \mu_p}.
\end{equation}

Extending to commensurate fractional derivatives (when all orders are $\alpha$) and utilizing a fractional comparison theorem, we find
\begin{equation}
N(t) = N_0 E_{\alpha, 1}(-\mu t^{\alpha}) + \left(\frac{\Lambda}{\mu}\right)^{\alpha} E_{\alpha, \alpha+1}(-\mu t^{\alpha}),
\end{equation}
where \(E_{\alpha, \beta}\) is  the asymptotic expansion of the Mittag-Leffler function. As \(t\) approaches infinity, this function's properties ensure \(N(t)\) approaches \((\frac{\Lambda}{\mu})^\alpha\), confirming that \(\Omega\) is positively invariant. Similarly, $\lim \sup_{t \rightarrow \infty} W(t) \leq \frac{(\eta_A^\alpha + \eta_S^\alpha) \Lambda^\alpha}{\mu^\alpha \mu_p^\alpha}$.
\end{proof}

It is important to note that our model incorporates incommensurate orders of derivatives, which introduces additional complexity to the analysis of system dynamics. Specifically, the traditional approach of equating the sum of compartmental populations to the total population may not hold precisely due to the asynchronous nature of the derivative orders. This deviation poses challenges in directly applying conventional methods to ascertain the boundedness and stability of solutions.

However, recent advancements in the field have shed light on handling such complexities. Under certain conditions, the exponential boundedness of solutions for systems with incommensurate fractional derivatives has been demonstrated \cite{diethelm2022asymptotic}. This suggests that, although the direct approach might falter, alternative strategies rooted in the latest mathematical frameworks can provide rigorous justification for the boundedness of solutions in our model. Such considerations are crucial for ensuring the robustness and biological fidelity of the model, especially when dealing with the behaviors exhibited by incommensurate order systems.

\subsection{Basic reproduction number}

The basic reproduction number is determined via the next-generation matrix method \cite{van}, applied to our model as specified in \eqref{model}. This approach focuses on calculating the rate at which new infections emerge within the infectious compartments, denoted as $\mathcal{X} = (E, I_A, I_S, W)$, and is described as follows:

\begin{equation*}
\mathcal{F(X)}=\begin{pmatrix}
\frac{\beta_1^{\alpha_E} S(t)W(t)}{1+\phi_1 W(t)}+\frac{\beta_2^{\alpha_E} S(t)\left(I_A(t)+ I_S(t)\right)}{1+\phi_2\left(I_A(t)+I_S(t)\right)}\\
0\\
0\\
0
\end{pmatrix},
\end{equation*}
and the rate of other transitions involving shedding compartment is obtained as follows\begin{equation*}
\mathcal{V(X)}=\begin{pmatrix}
(\psi^{\alpha_E} + \mu^{\alpha_E} + \omega^{\alpha_E})E\\
(\mu^{\alpha_{I_A}} + \sigma^{\alpha_{I_A}} + \gamma_A^{\alpha_{I_A}})I_A -(1-\delta)\omega^{\alpha_{I_A}} E\\
(\mu^{\alpha_{I_S}} + \sigma^{\alpha_{I_S}} + \gamma_S^{\alpha_{I_S}})I_S -\delta \omega^{\alpha_{I_S}} E\\
\mu_P^{\alpha_W} W -\eta_A^{\alpha_W} I_A - \eta_S^{\alpha_W} I_S.
\end{pmatrix}.
\end{equation*}
We identify a discrepancy in the original model \cite{mwalili} concerning the calculation of the basic reproduction number, \(R_0\). Specifically, the definitions of \(\mathcal{F}\) and \(\mathcal{V}\) in their work do not align with standard practices. To address this, we redefine $-\eta_A^{\alpha_W} I_A - \eta_S^{\alpha_W} I_S$ accurately as a factor in the transition of new infections, rather than as a contributor to the generation of new infections. Following this correction, the Jacobian matrices for \(\mathcal{F}\) and \(\mathcal{V}\) at the disease-free equilibrium (DFE) point, $\mathfrak{E}_0 = \left(\left(\frac{\Lambda}{\mu}\right)^{\alpha_S}, 0, 0, 0, 0, 0, 0\right)$, are derived as follows:

\begin{equation*}
\begin{split}
& J_{\mathcal{F}} = \left[ \begin{array}{cccc}
0& \beta_2^{\alpha_E}(\frac{\Lambda}{\mu})^{\alpha_S} & \beta_2^{\alpha_E}(\frac{\Lambda}{\mu})^{\alpha_S}  & \beta_1^{\alpha_E}(\frac{\Lambda}{\mu})^{\alpha_S} \\
0& 0 & 0&0\\
0&0& 0&0\\
0&0& 0&0
\end{array}
\right] \quad \text{ and } \\ &
J_{\mathcal{V}}= \left[ \begin{array}{cccc}
\psi^{\alpha_E} + \mu^{\alpha_E} + \omega^{\alpha_E
} & 0 & 0&0\\
-(1-\delta)\omega^{\alpha_{I_A}}& \mu^{\alpha_{I_A}} + \sigma^{\alpha_{I_A}} + \gamma_A^{\alpha_{I_A}} &0 &0\\
-\delta \omega^{\alpha_{I_S}}& 0& \mu^{\alpha_{I_S}} + \sigma^{\alpha_{I_S}} + \gamma_S^{\alpha_{I_S}} &0\\
0& -\eta_A^{\alpha_{W}}&-\eta_S^{\alpha_{W}} &\mu_P^{\alpha_{W}}
\end{array}
\right].
\end{split}
\end{equation*}
Finally, the basic reproduction number $R_0$ is obtained as the spectral radius of generation matrix $J_{\mathcal{F}}.J_{\mathcal{V}}^{-1}$, that is  precisely

\begin{equation}\label{R0}
\begin{split}
    R_0&= (\Lambda/\mu)^{\alpha_S}\left(\beta_2^{\alpha_E}\left(
    \frac{\delta \omega^{\alpha_{I_S}}}{\varpi_{is}\varpi_{e}}+ \frac{(1-\delta)\omega^{\alpha_{I_A}}}{\varpi_{ia}\varpi_{e}}\right)
    + \beta_1^{\alpha_E}  \left(\frac{\eta_S^{\alpha_W}\delta \omega^{\alpha_{I_S}}}{\mu_P^{\alpha_W}\varpi_{e}\varpi_{is}}+\frac{\eta_A^{\alpha_W}(1-\delta)\omega^{\alpha_{I_A}}}{\mu_P^{\alpha_W}\varpi_{e}\varpi_{ia}}\right)\right),
\end{split}
\end{equation}

where 
\[
\varpi_e = \psi^{\alpha_E} + \mu^{\alpha_E} + \omega^{\alpha_E} ,\quad \varpi_{ia}= \mu^{\alpha_{I_A}} + \sigma^{\alpha_{I_A}} + \gamma_A^{\alpha_{I_A}}, \quad \varpi_{is}= \mu^{\alpha_{I_S}} + \sigma^{\alpha_{I_S}} + \gamma_S^{\alpha_{I_S}}.
\]

\subsection{Local stability analysis}
In this section, we focus on establishing the local stability of the disease-free equilibrium point, denoted as $\mathfrak{E}_0$, for the model described in Equation \eqref{model}. This is achieved by analyzing the eigenvalues derived from the linearization matrix, specifically the system's Jacobian $J(\mathfrak{E}_0)$.

\begin{theorem}(\cite{LI20101810})\label{theo:2}
    Assuming the conditions $0<\alpha_i\leq 1$ for each $i=S, E, I_A, I_S, R, D, W$, let $M$ represent lowest common multiple of $r_i$ and $q_i$, where $\alpha_i=\frac{r_i}{q_i}$ with $r_i, q_i \in \mathbb{N}$, such that $gcd(r_i, q_i)=1, \forall i=S, E, I_A, I_S, R, D, W$. If every root $\lambda$ of the equation
    \begin{equation}
        \text{det}(\text{diag}(\lambda^{r_1},\lambda^{r_2},\lambda^{r_3},\lambda^{r_4},\lambda^{r_5},\lambda^{r_6},\lambda^{r_7})-J(\mathfrak{E}_0))=0,
    \end{equation}
    fulfills the condition $|\text{arg}(\lambda)|>\frac{\pi}{2M}$, then the system's equilibrium is locally asymptotically stable.
\end{theorem}

\begin{theorem}
    If $\alpha_E=\alpha_{I_A}=\alpha_{I_S}=\alpha_W=\alpha$ holds, the disease-free equilibrium point $\mathfrak{E}_0$ of model~\eqref{model} is locally asymptotically stable if the eigenvalues of $J(\mathfrak{E}_0)$ meet the criteria set forth in Theorem~\ref{theo:2}.
\end{theorem}
\begin{proof}
The Jacobian of Equation \eqref{model}, evaluated at the disease-free equilibrium $\mathfrak{E}_0$ is 

\[
J(\mathfrak{E}_0) = \begin{pmatrix}
-\mu^{\alpha_S}  & \psi^{\alpha_S} & \mathcal{A} & \mathcal{A} & 0 & 0 & \mathcal{B}  \\
0 & \mathcal{E} & \mathcal{C}  & \mathcal{C} & 0 & 0 & \mathcal{D} \\
0 & (1-\delta) \omega^{\alpha_{I_A}} & \mathcal{
F} & 0 & 0 & 0 & 0 \\
0 & \delta \omega^{\alpha_{I_S}} & 0 & \mathcal{
G} & 0 & 0 & 0\\
0 & 0 & \gamma_A^{\alpha_R} & \gamma_S^{\alpha_R} & -\mu^{\alpha_R} & 0 & 0 \\
0 & 0 & \sigma^{\alpha_D} & \sigma^{\alpha_D} & 0 & -\mu^{\alpha_D} & 0 \\
0 & 0 & \eta_A^{\alpha_W} & \eta_S^{\alpha_W} & 0 & 0 & -\mu_P^{\alpha_W} 
\end{pmatrix},
\]
where
$
\mathcal{A} = -\frac{\beta_2^{\alpha_S}\Lambda^{\alpha_S}}{\mu^{\alpha_S}},\quad
\mathcal{B} = -\frac{\beta_1^{\alpha_S}\Lambda^{\alpha_S}}{\mu^{\alpha_S}},\quad
\mathcal{C} = \frac{\beta_2^{\alpha_E}\Lambda^{\alpha_S}}{\mu^{\alpha_S}},\quad
\mathcal{D} = \frac{\beta_1^{\alpha_E}\Lambda^{\alpha_S}}{\mu^{\alpha_S}},\quad
\mathcal{E} = -(\psi^{\alpha_E} + \mu^{\alpha_E} +\omega^{\alpha_E}), \quad
\mathcal{F}=-(\mu^{\alpha_{I_A}} + \sigma^{\alpha_{I_A}} + \gamma_A^{\alpha_{I_A}}),\quad
\mathcal{G}=-(\mu^{\alpha_{I_S}} + \sigma^{\alpha_{I_S}} + \gamma_S^{\alpha_{I_S}}).
$

The characteristic equation derived from $J(\mathfrak{E}_0)$ simplifies to

\[\text{det}(\text{diag}(\lambda^{M\alpha_S},\lambda^{M\alpha_E},\lambda^{M\alpha_{I_A}},\lambda^{M\alpha_{I_S}},\lambda^{M\alpha_R},\lambda^{M\alpha_D},\lambda^{M\alpha_W})-J(\mathfrak{E}_0))=0\]
Given the assumption that $\alpha_E=\alpha_{I_A}=\alpha_{I_S}=\alpha_W=\alpha$ 
it follows that:
\begin{equation}\label{eq: char}
    (\lambda^{M\alpha_S}+\mu^{\alpha_S})(\lambda^{M\alpha_R}+\mu^{\alpha_R})(\lambda^{M\alpha_W}+\mu^{\alpha_W})(\lambda^{4M\alpha}+\mathcal{H}_1\lambda^{3M\alpha}+\mathcal{H}_2\lambda^{2M\alpha}+\mathcal{H}_3\lambda^{M\alpha}+\mathcal{H}_4)=0,
\end{equation}
where
\begin{equation*}
    \begin{split}
        \mathcal{H}_1 &=\mathcal{E}+\mathcal{F}+\mathcal{G}+\mu_P^\alpha \\
        \mathcal{H}_2 &=\mathcal{EF}+\mathcal{EG}+\mathcal{FG}+\mu_P^{\alpha}(\mathcal{E}+\mathcal{F+\mathcal{G}})-\mathcal{C}\omega^\alpha\\
        \mathcal{H}_3 &=\mathcal{EFG}-\omega^\alpha(\mathcal{C}(\mathcal{G}-\mu_P^\alpha+(\mathcal{G}-\mathcal{F})\delta)+\mathcal{D}\eta_S^\alpha)+\mathcal{D}\eta_A^\alpha\omega^\alpha(1-\delta)+\mu_P^\alpha(\mathcal{EF}+\mathcal{EG}+\mathcal{FG})\\
        \mathcal{H}_4 &= \mathcal{EFG}\mu_P^\alpha + \omega^\alpha(\mathcal{G}(1-\delta)(\mathcal{D}\eta_A^\alpha+\mathcal{C}\mu_P^\alpha)+\mathcal{F}\delta(\mathcal{D}\eta_S^\alpha - \mathcal{C}\mu_P^\alpha)).
    \end{split}
\end{equation*}

From the initial terms of Equation \eqref{eq: char}, specifically $(\lambda^{M\alpha_S}+\mu^{\alpha_S})=0$, we directly obtain:
\begin{equation*}
      \text{arg}(\lambda_{1,k})=\frac{\pi}{M\alpha_S}+2\frac{k\pi}{M\alpha_S}.
\end{equation*}
This leads us to conclude:
\begin{equation*}
    |\text{arg}(\lambda_{1,k})|>\frac{\pi}{2M}, \quad k=0, 1, ... , M\alpha_S-1. 
\end{equation*}
Likewise, examining the second and third terms of Equation \eqref{eq: char}, it follows that:
\begin{equation*}
\begin{split}
    |\text{arg}(\lambda_{2,k})|&>\frac{\pi}{2M}, \quad k=0, 1, ... , M\alpha_R-1,\\
    |\text{arg}(\lambda_{3,k})|&>\frac{\pi}{2M}, \quad k=0, 1, ... , M\alpha_W-1.
\end{split}
\end{equation*}
The remaining eigenvalues can be determined as the solutions to the equation
\begin{equation*}
\mathfrak{P}(\lambda^{M\alpha})= (\lambda^{M\alpha})^4+\mathcal{H}_1(\lambda^{M\alpha})^3+\mathcal{H}_2(\lambda^{M\alpha})^2+\mathcal{H}_3(\lambda^{M\alpha})+\mathcal{H}_4=0.
\end{equation*}

Based on Theorem \ref{theo:2}, it is required that \(|\text{arg}(\lambda)| > \frac{\pi}{2M}\). Given that \(M\alpha\) is a positive integer, this implies \(M\alpha|\text{arg}(\lambda)| > \frac{\alpha\pi}{2}\), which in turn implies \(|\text{arg}(\lambda^{M\alpha})| > \frac{\alpha\pi}{2}\). This satisfies the Routh-Hurwitz conditions, as elucidated in \cite{AHMED2007607, MATOUK20092166}, which are necessary and sufficient for the stability of the roots of the polynomial equation. We now proceed to define the discriminant of \(\mathfrak{P}(\lambda)\) as follows:

\[ \mathfrak{D}(\mathfrak{P}) = \begin{bmatrix}
1 & \mathcal{H}_1 & \mathcal{H}_2 & \mathcal{H}_3 & \mathcal{H}_4 & 0 & 0 \\
0 & 1 & \mathcal{H}_1 & \mathcal{H}_2 & \mathcal{H}_3 & \mathcal{H}_4 & 0 \\
0 & 0 & 1 & \mathcal{H}_1 & \mathcal{H}_2 & \mathcal{H}_3 & \mathcal{H}_4 \\
4 & 3\mathcal{H}_1 & 2\mathcal{H}_2 & \mathcal{H}_3 & 0 & 0 & 0 \\
0 & 4 & 3\mathcal{H}_1 & 2\mathcal{H}_2 & \mathcal{H}_3 & 0 & 0 \\
0 & 0 & 4 & 3\mathcal{H}_1 & 2\mathcal{H}_2 & \mathcal{H}_3 & 0 \\
0 & 0 & 0 & 4 & 3\mathcal{H}_1 & 2\mathcal{H}_2 & \mathcal{H}_3 \\
\end{bmatrix}. \]

Leveraging insights from \cite{AHMED2007607, MATOUK20092166}, the stability criteria for the equilibrium point $\mathfrak{E}_0$ are as follows:

\begin{enumerate}
    \item If
    \[
    \Delta_1 = \mathcal{H}_1, \quad \Delta_2 = \begin{vmatrix}
    \mathcal{H}_1 & 1 \\
    \mathcal{H}_3 & \mathcal{H}_2 \\
    \end{vmatrix}, \quad \Delta_3 = \begin{vmatrix}
    \mathcal{H}_1 & 1 & 0 \\
    \mathcal{H}_3 & \mathcal{H}_2 & \mathcal{H}_1 \\
    0 & \mathcal{H}_4 & \mathcal{H}_3 \\
    \end{vmatrix},
    \]
    then in the scenario where \(\alpha = 1\), the necessary and sufficient condition for the equilibrium point \(\mathfrak{E}_0\) to be locally asymptotically stable are
    \[
    \Delta_1 > 0, \Delta_2 > 0, \Delta_3 = 0, \mathcal{H}_4 > 0,
    \]
    and the above conditions are sufficient for \(\mathfrak{E}_0\) to be locally asymptotically stable for all \(\alpha \in [0, 1)\).
    
    \item If \(\mathfrak{D}({\mathfrak{P}}) > 0, \mathcal{H}_1 > 0, \mathcal{H}_2 < 0\) and \(\alpha > \frac{2}{3}\), then the equilibrium point \(\mathfrak{E}_0\) is unstable.
    
    \item If \(\mathfrak{D}({\mathfrak{P}}) < 0, \mathcal{H}_1 > 0, \mathcal{H}_2 > 0, \mathcal{H}_3 > 0, \mathcal{H}_4 > 0\) and \(\alpha < \frac{1}{3}\), then the equilibrium point \(\mathfrak{E}_0\) is locally asymptotically stable.
    
    Furthermore, if \(\mathfrak{D}(\mathfrak{P}) < 0, \mathcal{H}_1 < 0, \mathcal{H}_2 > 0, \mathcal{H}_3 < 0, \mathcal{H}_4 > 0\), then the equilibrium point \(\mathfrak{E}_0\) is unstable.
    
    \item If
    \(\displaystyle{\mathfrak{D}(\mathfrak{P}) < 0, \mathcal{H}_1 > 0, \mathcal{H}_2 > 0, \mathcal{H}_3 > 0, \mathcal{H}_4 > 0 \text{ and } \frac{\mathcal{H}_1 \mathcal{H}_4}{\mathcal{H}_2 \mathcal{H}_3} + \frac{\mathcal{H}_3}{\mathcal{H}_1} > \frac{\mathcal{H}_4}{\mathcal{H}_1},}
    \)
    then the equilibrium point \(\mathfrak{E}_0\) is locally asymptotically stable, for all \(\alpha \in (0, 1)\).
    
    \item \(\mathcal{H}_4 > 0\) is the necessary condition for the equilibrium point \(\mathfrak{E}_0\) to be locally asymptotically stable.
\end{enumerate}

\end{proof}

\subsection{Global stability analysis}
Here, we investigate the global stability of the steady state when the disease is dying out in the population.

\begin{theorem}
The disease-free equilibrium point ($\mathfrak{E}_0$) of system \eqref{model} is globally asymptotically stable whenever $R_0 \leq 1$.
\end{theorem}
\begin{proof}
Let us consider the following Lyapunov function:
\[
V(t)= b_0E(t)+b_1I_A(t)+b_2 I_S(t)+b_3W(t),
\]
where $b_1, \, b_1, \, b_2,$ and $b_3$ are positive constant to be determined.

By linearity of the Caputo derivative, we have
\[
^{C}D^{\alpha}_{0+}V(t)= b_0 {}^{C}D^{\alpha}_{0+}E(t)+b_1 {}^{C}D^{\alpha}_{0+}I_A(t)+b_2 {}^{C}D^{\alpha}_{0+}I_S(t)+b_3 {}^{C}D^{\alpha}_{0+}W(t).
\]
Next, by substituting expression of $^{C}D^{\alpha}_{0+}E(t)$, $^{C}D^{\alpha}_{0+}I_A(t)$, $^{C}D^{\alpha}_{0+}I_S(t)$ and $^{C}D^{\alpha}_{0+}W(t)$ from system model \eqref{model}, we obtain
\begin{align*}
^{C}D^{\alpha}_{0+}V(t)&=b_0\left( \frac{\beta_1^{\alpha_E} S(t)W(t)}{1+\phi_1 W(t)}+\frac{\beta_2^{\alpha_E}S(t)\left(I_A(t)+ I_S(t)\right)}{1+\phi_2\left(I_A(t)+I_S(t)\right)}-\psi^{\alpha_E} E(t) -\mu^{\alpha_E} E(t) -\omega^{\alpha_E} E(t) \right)\\
&+ b_1 \left( (1-\delta)\omega^{\alpha_{I_A}} E(t) - (\mu^{\alpha_{I_A}} + \sigma^{\alpha_{I_A}})I_A(t) - \gamma_A^{\alpha_{I_A}} I_A(t) \right) \\
& + b_2 \left( \delta \omega^{\alpha_{I_S}} E(t) - (\mu^{\alpha_{I_S}} + \sigma^{\alpha_{I_S}})I_S(t) -\gamma_S^{\alpha_{I_S}} I_S(t) \right)\\
& + b_3 \left( \eta_A^{\alpha_{W}} I_A(t) + \eta_S^{\alpha_{W}} I_S(t)-\mu_P^{\alpha_{W}} W(t) \right).
\end{align*}
Since, the inequality $S\leq (\frac{\Lambda }{\mu})^{\alpha_S}$ holds, it follows that
\begin{align*}
^{C}D^{\alpha}_{0+}V(t)&\leq
b_0\left( (\frac{\Lambda }{\mu})^{\alpha_S}\left(\frac{\beta_1^{\alpha_E} W(t)}{1+\phi_1 W(t)}+\frac{\beta_2^{\alpha_E}\left(I_A(t)+ I_S(t)\right)}{1+\phi_2\left(I_A(t)+I_S(t)\right)}\right)-\psi^{\alpha_E} E(t) -\mu^{\alpha_E} E(t) -\omega^{\alpha_E} E(t) \right)\\
& + b_1 \left( (1-\delta)\omega^{\alpha_{I_A}} E(t) - (\mu^{\alpha_{I_A}} + \sigma^{\alpha_{I_A}})I_A(t) - \gamma_A^{\alpha_{I_A}} I_A(t) \right) \\
& + b_2 \left( \delta \omega^{\alpha_{I_S}} E(t) - (\mu^{\alpha_{I_S}} + \sigma^{\alpha_{I_S}})I_S(t) -\gamma_S^{\alpha_{I_S}} I_S(t) \right)\\
& + b_3 \left( \eta_A^{\alpha_{W}} I_A(t) + \eta_S^{\alpha_{W}} I_S(t)-\mu_P^{\alpha_{W}} W(t) \right).
\end{align*}
Note that because parameters and state variables are positive, we have
\[
\displaystyle{\frac{1}{1+ \phi_1 W}\leq 1,\quad \textup{ and } \quad \frac{1}{1 + \phi_2 (I_A + I_S)}\leq 1.}
\]
It follows that
\begin{align*}
^{C}D^{\alpha}_{0+}V(t)&\leq
b_0\left( (\frac{\Lambda }{\mu})^{\alpha_S}\left(\beta_1^{\alpha_E} W(t)+\beta_2^{\alpha_E}(I_A(t)+ I_S(t))\right)-\psi^{\alpha_E} E(t) -\mu^{\alpha_E} E(t) -\omega^{\alpha_E} E(t) \right)\\
& + b_1 \left( (1-\delta)\omega^{\alpha_{I_A}} E(t) - (\mu^{\alpha_{I_A}} + \sigma^{\alpha_{I_A}})I_A(t) - \gamma_A^{\alpha_{I_A}} I_A(t) \right) \\
& + b_2 \left( \delta \omega^{\alpha_{I_S}} E(t) - (\mu^{\alpha_{I_S}} + \sigma^{\alpha_{I_S}})I_S(t) -\gamma_S^{\alpha_{I_S}} I_S(t) \right)\\ 
& + b_3 \left( \eta_A^{\alpha_{W}} I_A(t) + \eta_S^{\alpha_{W}} I_S(t)-\mu_P^{\alpha_{W}} W(t) \right).
\end{align*}

Rearranging and reducing lead to the following expression
\begin{align}\label{equistab}
^{C}D^{\alpha}_{0+}V(t)& \notag \leq
(b_2 \delta \omega^{\alpha_{I_S}} - b_0 \varpi_{e} - b_1 \omega^{\alpha_{I_A}} (\delta - 1)) E(t) \\ \notag
&+(b_3 \eta_A^{\alpha_W} - b_1 \varpi_{ia} + b_0 \beta_2^{\alpha_E} ( \frac{\Lambda}{\mu} )^{\alpha_S}) I_A(t) \\
&+(b_3 \eta_S^{\alpha_W} - b_2 \varpi_{is} + b_0 \beta_2^{\alpha_E} ( \frac{\Lambda}{\mu} )^{\alpha_S}) I_S(t) \\ \notag
&+(b_0 \beta_1^{\alpha_E} ( \frac{\Lambda}{\mu} )^{\alpha_S} - b_3 \mu_P^{\alpha_W}) W(t),
\end{align}
in which $\varpi_{is}=\gamma_S^{\alpha_{I_S}} + \mu^{\alpha_{I_S}} + \sigma^{\alpha_{I_S}}$, $\varpi_{ia}=\gamma_A^{\alpha_{I_A}} + \mu^{\alpha_{I_A}} + \sigma^{\alpha_{I_A}}$, and $\varpi_{e}=\mu^{\alpha_E} + \omega^{\alpha_E} + \psi^{\alpha_E}$.
Therefore, by choosing
\[
b_0= \mu_{P}^{\alpha_W}\varpi_{ia}\varpi_{is}\mu^{\alpha_S}, \quad b_1= \Lambda^{\alpha_S}(\beta_1^{\alpha_E} \eta_A^{\alpha_W} + \beta_2^{\alpha_E} \mu_{P}^{\alpha_W})\varpi_{is},\]
\[ b_2= \Lambda^{\alpha_S}(\beta_1^{\alpha_E} \eta_S^{\alpha_W} +  \beta_2^{\alpha_E} \mu_P^{\alpha_W})\varpi_{ia} \quad b_3=\Lambda^{\alpha_S} \beta_1^{\alpha_E} \varpi_{ia}\varpi_{is},
\]
it easy to see that $V$ is continuous and positive definite for all $E(t)>0$, $I_A(t)>0$, $I_S(t)>0$ and $W(t)>0$.
As a consequence, we obtain that
\begin{gather*}
b_3 \eta_A^{\alpha_W} - b_1 \varpi_{is} + b_0 \beta_2^{\alpha_E} ( \frac{\Lambda}{\mu} )^{\alpha_S}=0, \quad b_3 \eta_S^{\alpha_W} - b_2 \varpi_{ia} + b_0 \beta_2^{\alpha_E} ( \frac{\Lambda}{\mu} )^{\alpha_S}=0, \quad  b_0 \beta_1^{\alpha_E} ( \frac{\Lambda}{\mu} )^{\alpha_S} - b_3 \mu_P^{\alpha_W}=0,\\ \text{ and } \quad  b_2 \delta \omega^{\alpha_{I_S}} - b_0 \varpi_{e} - b_1 \omega^{\alpha_{I_A}} (\delta - 1)= \mu_{P}^{\alpha_W}\varpi_{ia}\varpi_{is}\mu^{\alpha_S} (R_0 -1).
\end{gather*}
Hence, putting altogether in the inequality \eqref{equistab}, we get
\[
^{C}D^{\alpha}_{0+}V(t) \leq \mu_{P}\varpi_{ia}\varpi_{is}\mu (R_0 -1).
\]
Finally, $^{C}D^{\alpha}_{0+}V(t) \leqslant 0$ if $R_0 \leqslant 1$. In addition, it is not hard to verify that the largest invariant set of $\left\{ \left(S, E, I_A, I_S, R, D, W \right) \in \mathbb{R}^{7}:\; ^{C}D^{\alpha}_{0+}V(t)=0 \right\}$ is the singleton $\{ \mathfrak{E}_0 \}$. Hence, by LaSalle's invariance principle~\cite{LASALLE196857}, we conclude that the disease-free equilibrium $\mathfrak{E}_0$ is globally asymptotically stable.

\end{proof}

Our examination of global stability spans the entire spectrum of $R_0$ values, albeit within the scope of commensurate orders. 
Global stability of the commensurate fractional form of model~\eqref{model} is confirmed through Ulam-Hyers stability~\cite{jung2011hyers}, with recent COVID-19 model applications in~\cite{cmc.2020.011508}.

To provide a foundation for the subsequent discussion, we introduce the inequality expressed as follows:
\begin{equation}\label{eq:10}
    |{}^{C}D^{\alpha}_{0+}\textbf{X}(t)-\textbf{F}(t, \textbf{X}(t))|\leq \epsilon, \quad t\in [0,T].
\end{equation}
We say a function $\Bar{\textbf{X}}\in \mathbb{R}^7_+$ is a solution of \eqref{eq:10} if and only if there exists $h \in \mathbb{R}^7_+$ satisfying

\begin{itemize}
    \item[1.] $|h(t)| \leq \epsilon$
    \item[2.] ${}^{C}D^{\alpha}_{0+}\Bar{\textbf{X}}(t)=\textbf{F}(t, \Bar{\textbf{X}}(t))+h(t), \quad t\in [0,T]$.
\end{itemize}

Notably, by applying Equation \eqref{eq11} alongside property 2 mentioned above, straightforward simplification reveals that any function $\bar{\mathbf{X}} \in \mathbb{R}^7_+$ meeting the conditions of Equation \eqref{eq:10} likewise fulfills the following associated integral inequality:

\begin{equation}\label{eq:32}
    |\Bar{\textbf{X}}(t)-\Bar{\textbf{X}}(0)- \frac{1}{\Gamma(\alpha)}\int_0^t(t-\tau)^{\alpha-1}\textbf{F}(\tau,\Bar{\textbf{X}}(\tau))| \leq \frac{T^\alpha}{\Gamma(\alpha+1)} \epsilon.
\end{equation}

Let $\mathfrak{F}=C([0,T];\mathbb{R})$ denote the Banach space of all continuous functions from $[0,T]$ to $\mathbb{R}$ equipped with the norm $\|\textbf{X}\|_{\mathfrak{F}}=\sup_{t\in [0,T]}\{|\textbf{X}|\}$, where $|\textbf{X}|=|S(t)|+|E(t)|+|I_A(t)|+|I_S(t)|+|R(t)|+|D(t)|+|W(t)|$.

The fractional order model \eqref{model} achieves Ulam-Hyers stability if there are some $\Sigma > 0$ ensuring that, for any given $\epsilon$, and for every solution $\Bar{\textbf{X}}$ meeting the conditions of \eqref{eq:10}, a corresponding solution $\textbf{X}$ to \eqref{short-model} can be found where 
\begin{equation}
    \|\Bar{\textbf{X}}(t) - \textbf{X}(t) \|_{\mathfrak{F}} \leq \Sigma \epsilon, \quad t \in [0,T].
\end{equation}
Moreover, this model is deemed to be generalized Ulam-Hyers stable if a continuous function $\Sigma_{F}: \mathbb{R}+ \to \mathbb{R}+$ exists, satisfying $\Sigma_F(0)=0$. This condition requires that, for any solution $\Bar{\textbf{X}}$ of \eqref{eq:10}, there must be a corresponding solution $\textbf{X}$ of \eqref{short-model} for which
\begin{equation}
    \|\Bar{\textbf{X}}(t) - \textbf{X}(t) \|_{\mathfrak{F}} \leq \Sigma_F \epsilon, \quad t \in [0,T].
\end{equation}
We proceed to detail the stability results for the fractional order model.

\begin{theorem}
Assuming the conditions and conclusions of Lemma~\ref{lemm: lips} and Theorem \ref{theo: 1} are satisfied, i.e. $\Sigma \frac{T^\alpha}{\Gamma(\alpha+1)}<1$, it follows that the model specified in \eqref{short-model} exhibits generalized Ulam-Hyers stability.
\end{theorem}

\begin{proof}
    Given that $\textbf{X}$ is a unique solution to \eqref{short-model} confirmed by Lemma~\ref{lemm: lips} and Theorem \ref{theo: 1}, and $\Bar{\textbf{X}}$ meets the criteria of \eqref{eq:10}, reference to equations \eqref{eq11} and \eqref{eq:32} leads us to conclude that for any $\epsilon > 0$ and $t \in [0, T]$, the following relationship holds:
    \begin{equation}
\begin{split}
\| \Bar{\textbf{X}} - \textbf{X} \|_{\mathfrak{F}}  & = \sup_{t \in [0,T]} |\Bar{\textbf{X}} - \textbf{X}| \\
& = \sup_{t \in [0,T]} \left| \Bar{\textbf{X}} - {\textbf{X}}_0 - \frac{1}{\Gamma(\alpha)} \int_0^t (t - \tau)^{\alpha-1} \textbf{F}(t, \textbf{X}(\tau)) d\tau \right| \\
& \leq \sup_{t \in [0,T]} \left| \Bar{\textbf{X}}(t) - \Bar{\textbf{X}}_0 - \frac{1}{\Gamma(\alpha)} \int_0^t (t - \tau)^{\alpha-1} \textbf{F}(t, \Bar{\textbf{X}}(\tau)) d\tau \right| \\
& \quad + \sup_{t \in [0,T]}  \frac{1}{\Gamma(\alpha)} \int_0^t (t - \tau)^{\alpha-1} \left| \textbf{F}(t, \Bar{\textbf{X}}(\tau)) - \textbf{F}(t, \textbf{X}(\tau)) \right| d\tau  \\
& \leq \frac{\epsilon T^\alpha}{\Gamma(\alpha+1)} + \frac{\Sigma}{\Gamma(\alpha)} \sup_{t \in [0,T]} \int_0^t (t - \tau)^{\alpha-1} |\Bar{\textbf{X}}(\tau) - {\textbf{X}}(\tau)| d\tau \\
& \leq \frac{\epsilon T^\alpha}{\Gamma(\alpha+1)} + \frac{\Sigma T^\alpha}{\Gamma(\alpha+1)} \|\Bar{\textbf{X}}(\tau) - {\textbf{X}}(\tau)\|_{\mathfrak{F}}.
\end{split}
\end{equation}
From this, we derive that $\|\Bar{\textbf{X}} - \textbf{X}\|_{\mathfrak{F}} \leq \Sigma_F \epsilon$, with $\Sigma_F$ defined as $\frac{T^{\alpha}}{\Gamma(\alpha+1)-T^\alpha \Sigma}$.
\end{proof}

\section{Numerical results}
\label{sec:numerical}

In this section, we demonstrate the capability of Model \eqref{model} to simulate the dynamics of COVID-19 transmission. This includes consideration for environmental contamination by infected individuals. Additionally, we compare the basic reproduction number, $R_0$, and its sensitivity across various scenarios and model configurations. The model is fine-tuned using data collected from South Africa, obtained from Johns Hopkins University's Center for Systems Science and Engineering (CSSE)~\cite{DataCSSE}. 

South Africa confirmed its first case of COVID-19 in March 2020, prompting a National State of Disaster to be declared on March 15. A nationwide lockdown was then imposed on March 26. To better understand the spread of the disease in South Africa, our study focuses on simulating its propagation from the start of the lockdown measures until the end of the first peak of the epidemic, which occurred on September 22, 2020. In light of the United Nations documentation, as reported by the Department of Economic and Social Affairs, Population Division, World Population Prospects 2022, we incorporate a birth rate of $\Lambda=19.995/1000$ (19.995 births per 1000 individuals) and a natural human death rate of $\mu=9.468/1000$, while also assuming the initial conditions $I_S(0)=17, R(0)=0$, and $D(0)=0$, based on the early phase of the epidemic.

\begin{figure}[ht!]
    \centering
    \includegraphics[width=1\textwidth]{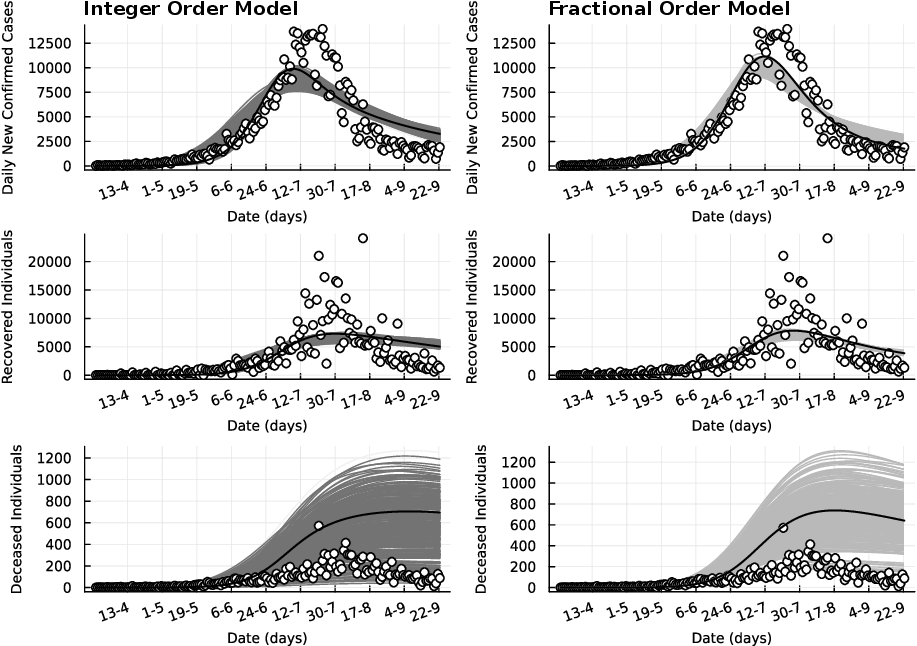}
    \caption{Model Accuracy Assessment: It demonstrates the accuracy of the model \eqref{model}, with integer order derivatives (left panels) and fractional order derivatives (right panels), in fitting the daily new confirmed cases, recovered individuals, and deceased individuals, obtained from CSSE~\cite{DataCSSE}. The circles represent the data points. The grey curves show the model's results with estimated parameter values and initial conditions, and the optimal fit across all three categories, evaluated via root mean square deviation (RMSD), is depicted with black curves.}
        \label{fig:fit}
\end{figure}

\begin{figure}[ht!]
    \centering
    \includegraphics[width=.5\textwidth]{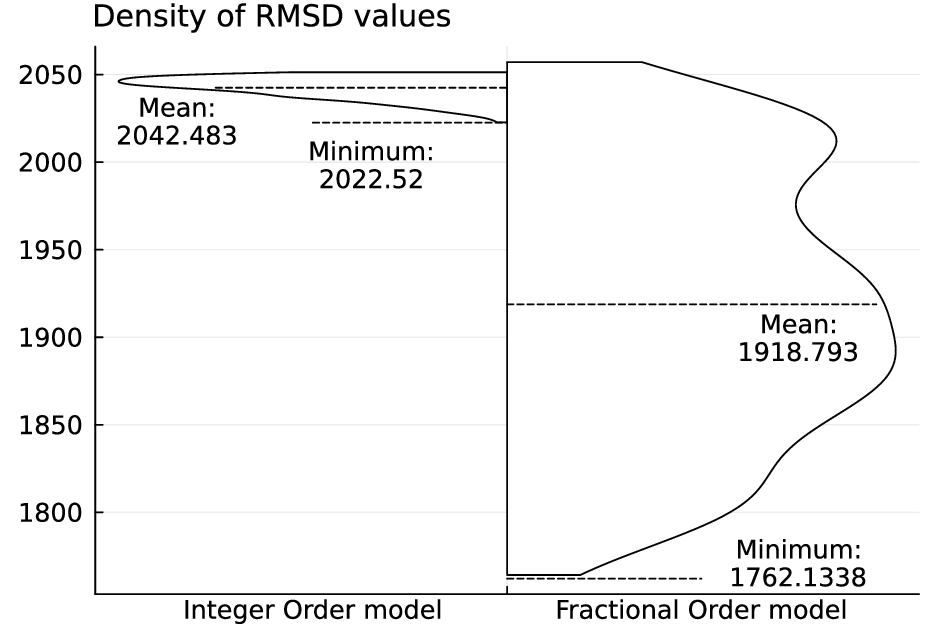}
    \caption{Comparative Error Score Distribution: It illustrates the distribution of root mean square deviation (RMSD) comparing real data with the fitted models using integer (left) and fractional (right) orders.}
        \label{fig:errors}
\end{figure}

We determine the remaining parameter values and initial conditions by fitting model \eqref{model} to the available data: daily new confirmed cases, recovered individuals, and deceased individuals. 
For parameters \(\mu_p\), \(\phi_1\), \(\phi_2\), \(\beta_1\), \(\beta_2\), \(\delta\), \(\psi\), \(\omega\), \(\sigma\), \(\gamma_S\), \(\gamma_A\), \(\eta_S\), \(\eta_A\), we assume truncated normal distributions. These parameters must fall within a certain range of 0 to 1. The prior for the initial susceptible population, \(S(0)\), is set to a truncated normal distribution with a wide range (1000 to 1500000) and large deviation, reflecting the variability in population sizes for different regions or countries. The initial conditions for exposed, \(E(0)\), asymptomatic infectious, \(I_A(0)\), and pathogen presence, \(W(0)\), compartments are also set to truncated normal distributions. These distributions have ranges reflecting possible small initial values (0 to 300), based on the early phase of an epidemic when these compartments start to grow.

We optimize the order derivatives to examine the efficacy of fractional orders in the fitting process. We utilize the root mean square deviation (RMSD) as a measure of the model accuracy, where RMSD$(x,\widehat{x})=\sqrt{\frac{1}{n}\sum_{t=1}^n(x_t-\widehat{x}_t)^2}$. Here, $n$ denotes the number of data points, $x$ denotes the estimated values, and $\widehat{x}$ denotes the actual values. Notably, considering biologically or clinically meaningful parameters plays a crucial role when fitting a model to data. For instance, previous research studies \cite{NGONGHALA2020108364, ferguson2020report} suggest that the parameters $\gamma_S$ and $\gamma_A$ should be confined to the range of (0,1). Given this point, we have chosen to display a set of fitted results, rather than solely presenting the best fit, to demonstrate the overall efficiency of the model utilizing fractional calculus.

In Figure~\ref{fig:fit}(left panels), the Bayesian inference method has been employed to provide initial estimates for both the parameter values and initial conditions, as the results depicted by the grey curves and the black curve represents the best fit with an RMSD of 2022.52. In Figure~\ref{fig:fit}(right panels), we have further optimized the order derivatives, using the obtained parameters and conditions. The grey curves depict their results and the black curve shows the best fit with an RMSD of 1762.13. 

Figure~\ref{fig:errors} depicts the density of errors for the fitted parameters and order derivatives, categorized into integer and fractional orders, represented by the left and right sides of the violin plot, respectively. The results reveal that the mean and minimum errors of the model with fractional orders are comparatively lower than those of the model with integer orders. It highlights the enhanced flexibility and improved data fitting by incorporating fractional derivatives into the model. 

The tables presenting the optimal values of the parameters, initial conditions, and order derivatives, along with their corresponding statistical properties such as mean, standard deviation, and median can be found in Table~\ref{tab:Sensitivity}. 

\begin{table}[th!]
\centering
\caption{Statistical Properties of Model Parameters: This table presents the mean, standard deviation, and median of model parameters, order derivatives, and initial conditions, alongside the sensitivity index for $R_0$ ($\mathcal{S}_{\mathcal{P}}^{R_0}$) considering optimized values within integer (ODE) and fractional order derivatives (FDE) models.}
\begin{tabular}{c|ccc|cccc}
\hline
Parameters & Mean & STD & Median & \multicolumn{2}{c}{Optimized value} &  \multicolumn{2}{c}{Sensitivity}  \\
 &      &     &        &     ODE & FDE       & ODE & FDE \\ \hline
$\Lambda $ &  - & -& -& - & - & 1.0000& 1.0000\\
$\mu$ & - & -& -& - & - &  -0.7109& -0.7115\\
$\mu_p$ & 0.0038   & 0.0039 &    0.0026   &     0.0036 &0.0023& -0.2657& -0.2657 \\
$\phi_1$ &  3.11e-6 & 1.32e-5 & 5.55e-7  &   8.15e-7 &1.04e-6 & 0.0000&0.0000\\
$\phi_2$ & 0.6341  &  0.2381   &  0.6711 &    0.4394&0.8344 & 0.0000& 0.0000\\
$\beta_1$ &  1.32e-5 & 4.77e-5  & 2.94e-6  & 2.80e-6&2.32e-5 & 0.9964 & 0.7250\\
$\beta_2$ & 1.18e-6  & 1.20e-6  & 8.23e-7 & 8.86e-7&9.05e-7 & 0.8618 & 0.0025\\
$\delta$ &0.3270  &  0.0633  &  0.3335   &  0.4947&0.3880 & -0.9737&-0.9757\\
$\psi$ &  0.0996  &  0.1703  &  0.0323  & 0.0002& 0.0032  & -0.0004& -0.0004\\
$\omega$ & 0.4764  &  0.2060   &  0.4342  & 0.5396&0.7022  & 0.0173 & 0.0577\\
$\sigma$ & 0.0015  &  0.0005  &  0.0015   & 0.0012&0.0017 & -0.0591 & -0.0591\\
$\gamma_S$ & 0.0009  & 0.0007  & 0.0007  & 0.0701&0.0009 & -0.0021 & -0.0010\\
$\gamma_A$ & 0.0726  &  0.0112   &  0.0716   &   0.0001&0.0744 & -0.0074& -0.0074\\
$\eta_S$ &  0.0065  &  0.0099  &   0.0030    & 0.0064&0.0007 & 0.0023 & 0.0014\\
$\eta_A$ & 0.4177   & 0.2867   & 0.3542 &  0.3585&0.0524 & 0.9941 &0.9950\\
\hline
Order derivatives & Mean & STD & Median & \multicolumn{2}{c}{Optimized value}& \multicolumn{2}{c}{Sensitivity}\\
\hline
$\alpha_S$ &  1.0000 & 1.1e-16 & 1.0000 & \multicolumn{2}{c}{1.0000}&\multicolumn{2}{c}{0.7504} \\
$\alpha_E$ & 0.9702 & 0.0157 & 0.9671 & \multicolumn{2}{c}{0.9591}&\multicolumn{2}{c}{-10.943} \\
$\alpha_{I_A}$ & 0.9853 & 0.0530 & 1.0000 & \multicolumn{2}{c}{1.0000}&\multicolumn{2}{c}{4.4178}\\
$\alpha_{I_S}$ &  0.8641 & 0.0783 & 0.8604 & \multicolumn{2}{c}{0.7868}&\multicolumn{2}{c}{0.0035}\\
$\alpha_{R}$ & 0.7847 & 0.0776 & 0.7665 & \multicolumn{2}{c}{0.7000}&\multicolumn{2}{c}{0.0000}\\
$\alpha_D$ & 0.7731 & 0.1171 & 0.7000 & \multicolumn{2}{c}{0.7018}&\multicolumn{2}{c}{0.0000}\\ 
$\alpha_W$ & 0.9990 & 0.0090 & 1.0000 & \multicolumn{2}{c}{1.0000}&\multicolumn{2}{c}{4.6752}\\
\hline
Initial conditions &  Mean & STD & Median & \multicolumn{4}{c}{Optimized value} \\
 &      &     &        &     \multicolumn{2}{c}{ODE}      &   \multicolumn{2}{c}{FDE}  \\ 
\hline
$S(0)$ &  70088.0 & 16318.8 & 65858.9 & \multicolumn{2}{c}{57674.0} & \multicolumn{2}{c}{83379.9}\\
$E(0)$ & 0.81708 & 0.89169 & 0.57169 & \multicolumn{2}{c}{0.1181} & \multicolumn{2}{c}{0.1002}\\
$I_A(0)$ & 0.38595 & 0.47730 & 0.20705 & \multicolumn{2}{c}{0.1093} & \multicolumn{2}{c}{0.0404}\\
$W(0)$ &  0.64895 & 1.25392 & 0.30470 & \multicolumn{2}{c}{0.0015} & \multicolumn{2}{c}{1.0605}\\
\hline
\end{tabular}
\label{tab:Sensitivity}
\end{table}

\subsection{Sensitivity analysis of the basic reproduction number}
This section examines the influence of different factors on the transmission of epidemic diseases. Policymakers can make informed decisions regarding public health interventions by assessing the sensitivity of $R_0$ to various parameters, such as social distancing measures. Furthermore, this methodology allows for identifying parameters that have the most significant impact on disease propagation. To assess the effects of minor variations in a parameter or a derivative order, denoted both by $\mathcal{P}$, on the basic reproduction number $R_0$, we introduce the forward normalized sensitivity index of $R_0$ for $\mathcal{P}$, which can be mathematically formulated as:
\begin{equation*}
    \mathcal{S}_{\mathcal{P}}^{R_0}=\frac{\partial R_0}{\partial \mathcal{P}}\frac{\mathcal{P}}{R_0}.
\end{equation*}

\begin{figure}[ht!]
    \centering
    \includegraphics[width=.9\textwidth]{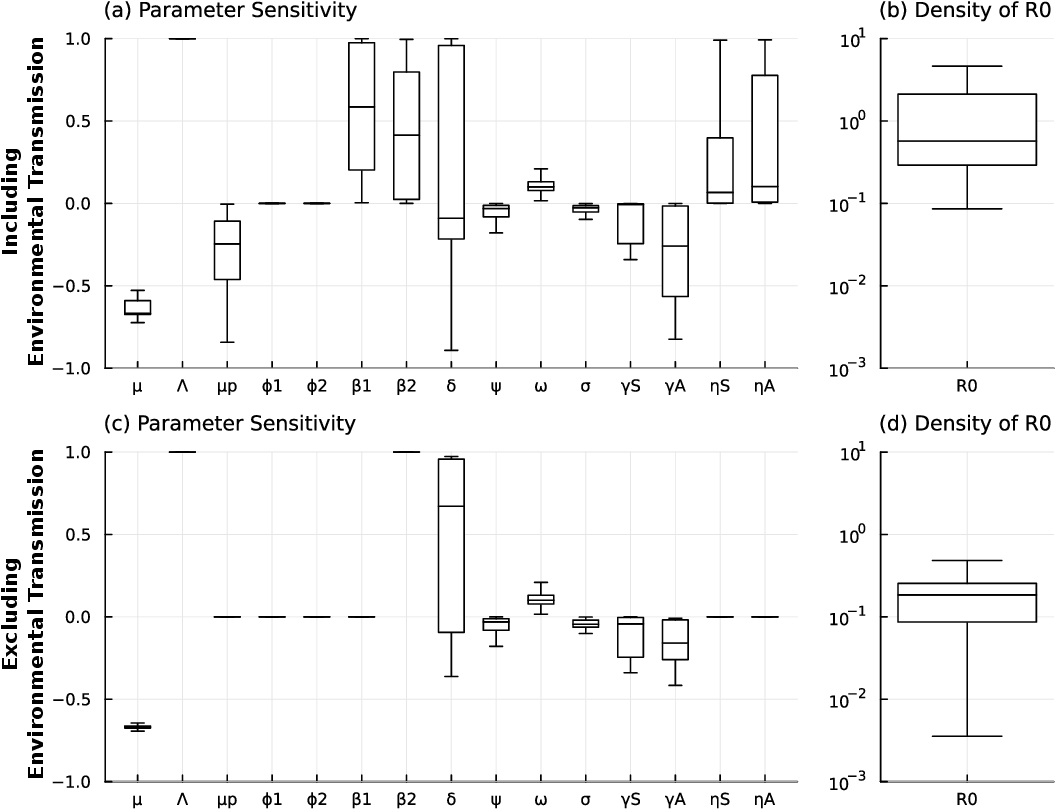}
    \caption{Sensitivity Analysis of $R_0$ for Model~\eqref{model} with Integer Order Derivatives. (a) The boxplots of the sensitivity indices of $R_0$ to the parameter values, and (b) the distribution of $R_0$ for these values. (c) The sensitivity indices of $R_0$ to parameters when the environmental pathogen-related coefficients $\beta_1$, $\eta_S$, and $\eta_A$ are constrained to zero. (d) Illustration of the resulting influence on $R_0$ when the contributions of environmental transmission factors are eliminated.}
    \label{fig:SenR0}
\end{figure}

\begin{figure}[ht!]
    \centering
    \includegraphics[width=.9\textwidth]{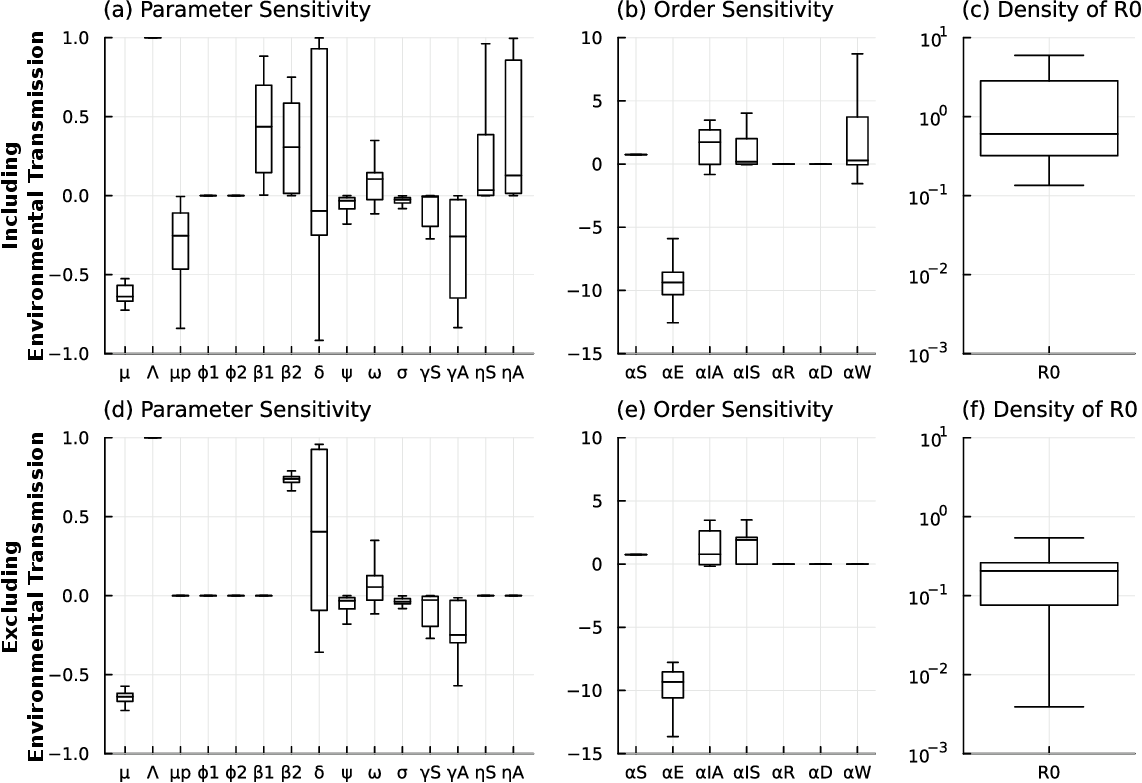}
    \caption{Sensitivity Analysis of $R_0$ for Model~\eqref{model} with Fractional Order Derivatives. (a) The boxplots of the sensitivity indices of $R_0$ to the parameter values and (b) order derivatives, and (c) the distribution of $R_0$ for these values. (d) The sensitivity indices of $R_0$ to parameters and (e) order derivatives when the environmental pathogen-related coefficients $\beta_1$, $\eta_S$, and $\eta_A$ are constrained to zero. (f) Illustration of the resulting influence on $R_0$ when the contributions of environmental transmission factors are eliminated.}
    \label{fig:SenR0Frac}
\end{figure}

If the sensitivity index of $\mathcal{S}_{\mathcal{P}}^{R_0}$ is positive, it indicates an increase in the value of the basic reproduction number $R_0$ concerning $\mathcal{P}$. Conversely, the sensitivity index will be negative if the value of $R_0$ decreases in response to $\mathcal{P}$ changes. 

The computed sensitivity indices for $R_0$ to the optimized values have been demonstrated in Table \ref{tab:Sensitivity}. In addition, we illustrate the distribution of sensitivity indices of the parameters on $R_0$ through boxplots, together with the density of $R_0$ for these values in Figure~\ref{fig:SenR0}(a-b), for the integer order model. Our analysis revealed that only the parameter $\delta$, representing the proportion of symptomatic infectious individuals, mainly exhibited positive but sometimes negative sensitivity. In contrast, the other parameters maintained their sign of sensitivity indices. 

Moreover, exploring the sensitivity of \(R_0\) to parameters that facilitate virus transmission---specifically \(\beta_1\), \(\eta_S\), and \(\eta_A\)---presents an interesting area of discussion. Notice that \(\mu_P\), representing the natural death rate of the pathogen in the environment, is excluded from this analysis due to its consistently negative sensitivity index and it is beyond human control. Similarly, \(\phi_1\) is omitted as it does not feature in the equation \eqref{R0} and thus has no impact on the sensitivity. In contrast, the sensitivity indices of \(\beta_1\), \(\eta_S\), and \(\eta_A\) are positive, raising interesting considerations regarding the potential effects on \(R_0\) through rigorous preventative measures, such as enhanced caution, hand washing, mask usage, and efforts to eliminate environmental contamination. Under these assumptions, the significant impact on \(R_0\) is illustrated in Figure~\ref{fig:SenR0}(c-d), demonstrating that \(R_0\) could be reduced to below 1 with optimal interventions.

A particularly significant and intriguing aspect is the comparison of the impacts of fractional orders in sensitivity analysis. Figure~\ref{fig:SenR0Frac}(a) illustrates that in the fractional model, all parameters exhibit somewhat similar behaviors in terms of their sensitivity indices on $R_0$ when compared to the integer model. However, there are instances where $\omega$ may negatively influence $R_0$. Figure~\ref{fig:SenR0Frac}(b) unveils the sensitivity of the order of derivatives on $R_0$, which is noteworthy due to the challenging nature of interpreting the effects of orders on $R_0$. Interestingly, it reveals that the sensitivity index of $\alpha_E$ is negative, whereas the indices for other orders are predominantly positive. This provides a preliminary insight into the impact of derivative orders on $R_0$. It also results that $R_0$ values in fractional-order models are marginally higher than those in integer-order models, comparing Figure~\ref{fig:SenR0Frac}(c) with Figure~\ref{fig:SenR0}(c).
Furthermore, echoing the observations made with the integer model, Figure~\ref{fig:SenR0Frac}(d-f) also demonstrates how mitigating environmental transmission factors influences the outcomes in the fractional order model.

\subsection{Numerical methods and implementation}
The numerical analyses in this study were carried out using the Julia programming language. The FdeSolver.jl package (v 1.0.7) solves fractional differential equations, implementing predictor-corrector algorithms and product-integration rules~\cite{FdeSolver}. Parameter estimation was performed using Bayesian inference and Hamiltonian Monte Carlo (HMC) with the Turing.jl package, and ODEs were solved using the DifferentialEquations.jl package. The order of derivatives was optimized using the function (L)BFGS, which employs the (Limited-memory) Broyden–Fletcher–Goldfarb–Shanno algorithm from the Optim.jl package in combination with FdeSolver.jl.

In our efforts to derive parameters and orders, our goal is not merely to identify the optimal solutions, parameters, and orders, but rather to establish a reasonable range of values. This approach allows us to assess the impact of pathogen shedding in our modeling and to explore the influence of fractional derivatives within this context. However, this process necessitates certain considerations. For instance, we account for a coefficient representing the proportion of asymptomatic cases (\(I_A\)) that undergo COVID-19 testing (denoted by \(T\)), aiming for a more realistic evaluation of confirmed case data. This is based on the premise that data is predominantly derived from those who have been tested for COVID-19. Assuming that all symptomatic cases (\(I_S\)) and only a portion of \(I_A\) are tested, the confirmed cases are represented by \(I_S + T I_A\). This adjustment introduces an additional equation to more accurately fit the recovered (\(R\)) compartment, acknowledging that records of recovery only include tested individuals from both \(I_A\) and \(I_S\) groups. It is important to note that while this modification does not alter the overall model or our analysis---since it involves independent equations---it enhances the realism of the fit.

Ultimately, one must recognize that the flexibility afforded by the fractional model does not inherently equate to a superior or more accurate fit, as it could potentially lead to overfitting. Therefore, further analysis, such as cross-validation, is essential, especially in scenarios where parameter values have significant implications for policy-making.

\subsection{Data and code availability}
All computational results for this article are available on GitHub and accessible via \url{https://github.com/moeinkh88/Covid_Shedding}, including all data and code used for running the simulations and generating the figures.

\section{Conclusion}
\label{sec:conc}

In concluding the paper on the fractional model of COVID-19 transmission with environmental pathogens as a shedding effect, it is pertinent to encapsulate the significant contributions, findings, and potential implications for future research and policy-making. This paper has meticulously developed a fractional-order compartmental model, enhancing the conventional framework by incorporating environmental transmission vectors and accounting for asymptomatic and symptomatic infections. Through rigorous mathematical analysis, including existence, uniqueness, boundedness of solutions, and stability analyses, the model's robust theoretical foundation has been solidly established.

Numerical results have significantly demonstrated the model's utility in simulating the dynamics of COVID-19 transmission, highlighting a more flexible fit with actual data when fractional derivatives are employed, thus offering broader insight into the disease's spread. The sensitivity analysis of the basic reproduction number $R_0$ has revealed critical parameters and order of derivatives influencing the transmission dynamics, guiding targeted interventions. The model has also highlighted the considerable impact of environmental transmission factors, suggesting that mitigating these can substantially lower $R_0$, thereby controlling the outbreak's spread.

This study emphasizes the importance of incorporating environmental factors and fractional calculus in modeling infectious diseases, offering a more comprehensive and flexible framework. Recognizing the combined and cross-disciplinary endeavor, this research sets a precedent for subsequent inquiries, encouraging the academic and scientific arenas to persist in evolving and refining mathematical models to confront infectious diseases more effectively in the future.

\section*{Declaration of Competing Interest}
The authors declare no conflict of interest.

\section*{Acknowledgment}
This study has been supported by the Academy of Finland (330887 to MK, LL), the European Union’s Horizon 2020 research and innovation program (952914 to LL), and the UTUGS graduate school of the University of Turku (to MK).

The authors wish to acknowledge CSC–IT Center for Science, Finland, for computational resources.

\bibliographystyle{unsrt}
\bibliography{reference}

\end{document}